\documentclass[runningheads,a4paper]{llncs}

\usepackage[unicode,pdfborder={0 0 0}]{hyperref}
\usepackage[all]{hypcap}
\usepackage{amsmath}
\usepackage{amssymb}
\usepackage{url}
\usepackage{graphicx}
\usepackage{tikz}
\usetikzlibrary{chains}
\tikzset{every picture/.style={thick},every node/.style={draw,circle}}

\newcommand{\keywords}[1]{\par\addvspace\baselineskip
\noindent\keywordname\enspace\ignorespaces#1}
{\catcode`,=\active\xdef,{\mathchar\the\mathcode`\,\penalty1000\mskip0muplus3mu\relax}}
\mathcode`,="8000
\DeclareMathOperator\enc{enc}

\begin{document}

\mainmatter

\title{On NP-Hardness of the Paired de Bruijn Sound Cycle Problem}

\author{Evgeny Kapun \and Fedor Tsarev}

\institute{St. Petersburg National Research University of Information\\
Technologies, Mechanics and Optics\\
Genome Assembly Algorithms Laboratory\\
197101, Kronverksky pr., 49, St. Petersburg, Russia\\
\url{tsarev@rain.ifmo.ru}\\
\url{http://genome.ifmo.ru/}}

\maketitle

\begin{abstract}

The paired de Bruijn graph is an extension of de Bruijn graph incorporating mate pair information for genome assembly proposed by Mevdedev et al. However, unlike in an ordinary de Bruijn graph, not every path or cycle in a paired de Bruijn graph will spell a string, because there is an additional soundness constraint on the path. In this paper we show that the problem of checking if there is a sound cycle in a paired de Bruijn graph is NP-hard in general case. We also explore some of its special cases, as well as a modified version where the cycle must also pass through every edge.

\end{abstract}

\section{Introduction}

Current genome sequencing technologies rely on the shotgun method --- the genome is split into several small fragments which are read directly. Some of the technologies generate single reads, while others generate mate-pair reads --- genome fragments are read from both sides. The problem of reconstructing the initial genome from these small fragments (reads) is known as the genome assembly problem. It is one of the fundamental problems of bioinformatics. Several models for genome assembly were studied by researchers.

One of the models for the single reads case is based on the maximum parsimony principle --- the original genome should be the shortest string containing all reads as substrings. This leads to the Shortest Common Superstring (SCS) problem which is NP-hard \cite{scs}. In the de Bruijn graph model proposed in \cite{dbg} each read is represented by a walk in the graph. Any walk containing all the reads as subwalks represents a valid assembly. Consequently, the genome assembly problem is formulated as finding the shortest superwalk. This problem, known as Shortest De Bruijn Superwalk problem (SDBS), was shown to be NP-hard \cite{wabi}.

In \cite{mlga} an algorithm for reads' copy counts estimation based on maximum likelihood principle was proposed. A similar algorithm can be applied to find multiplicities of the de Bruijn graph edges, so, the De Bruijn Superwalk with Multiplicities problem (DBSM) can be formulated. This problem have been proven to be NP-hard as well \cite{kapun-tsarev}.

Paired-end reads case is much less studied. To the best of our knowledge the only model which deals with paired-end reads is the paired de Bruijn graph proposed in \cite{paired-dbg}. However, not every path or cycle in a paired de Bruijn graph corresponds to a correct genome assembly, because there is an additional soundness constraint on the walk. Computational complexity for the problem of finding a sound cycle in the paired de Bruijn graph remained unknown \cite{pham-recomb-ab}. In this paper we show that this problem is NP-hard.

\section{Definitions}

A \emph{de Bruijn graph} of order \(k\) over an alphabet \(\Sigma\) is a directed graph in which every vertex has an associated label (a string over \(\Sigma\)) of length \(k\) and every edge has an associated label of length \(k+1\). All labels within a graph must be distinct. If an edge \((u,v)\) has an associated label \(l\), then the label associated with \(u\) must be a prefix of \(l\) and the label associated with \(v\) must be a suffix of \(l\).

Every path in a de Bruijn graph spells a string. A string spelled by a path \(v_1\), \(e_1\), \(v_2\), \ldots, \(e_{n-1}\), \(v_n\) of length \(n\) is a unique string \(s\) of length \(n+k-1\) such that the label associated with \(v_i\) occurs in \(s\) at position \(i\) for all \(1\le i\le n\), and the label associated with \(e_i\) occurs in \(s\) at position \(i\) for all \(1\le i\le n-1\). Every cycle of length \(n\) in a de Bruijn graph spells a cyclic string of length \(n\) having the same properties.

In a \emph{paired de Bruijn graph} each vertex and each edge has an associated \emph{bilabel} instead of a label. A bilabel is an ordered pair of strings of the same length (equal to the order of the graph), denoted as \((a,b)\). We say that \((a_1,b_1)\) is a prefix of \((a_2,b_2)\) iff \(a_1\) is a prefix of \(a_2\) and \(b_1\) is a prefix of \(b_2\). Suffix is defined analogously. As in ordinary de Bruijn graphs, all bilabels must be distinct, however, individual labels of which bilabels consist may coincide.

Similarly to the ordinary de Bruijn graph, every path in a paired de Bruijn graph spells a pair of strings, and every cycle spells a pair of cyclic strings. We say that a pair of strings \((s_1s_2\ldots s_n,t_1t_2\ldots t_n)\) of length \(n\) \emph{matches with shift \(d\)} iff \(s_{i+d}=t_i\) for all \(1\le i\le n-d\). Analogously, a pair of cyclic strings \((s_1s_2\ldots s_n,t_1t_2\ldots t_n)\) matches with shift \(d\) iff \(s_{i+d}=t_i\) for all \(1\le i\le n-d\) and \(s_i=t_{i+n-d}\) for all \(1\le i\le d\).

We say that a path in a paired de Bruijn graph is \emph{sound with respect to shift \(d\)}, or just \emph{sound}, iff the pair of strings it spells matches with shift \(d\). We say that a cycle in a paired de Bruijn graph is sound iff the pair of cyclic strings matches with shift \(d\).

We say that a path or a cycle is \emph{covering} if it includes all the edges in a graph. We say that a set of paths or cycles covers the graph iff every edge of the graph belongs to at least one path or cycle in the set.

A \emph{promise problem} is a kind of decision problem where only inputs from some set of valid inputs are considered. Specifically, a promise problem is defined by a pair of disjoint sets \((S_+,S_-)\). A solution to the problem is a program which outputs ``yes'' when run on inputs in \(S_+\) and outputs ``no'' when run on inputs in \(S_-\). However, when run on inputs outside of \(S_+\cup S_-\), its behavior may be arbitrary: it may return any result, exceed its allowed time and memory bounds, or even hang.

Note that a promise problem \((S_+,S_-)\) is at most as hard as \((S_+',S_-')\) if \(S_+\subseteq S_+'\) and \(S_-\subseteq S_-'\), because the solution for the latter problem would solve the former problem as well. Particularly, \((S_+',S_-')\) is NP-hard if \((S_+,S_-)\) is NP-hard. Also, an ordinary decision problem defined by set \(S\) is the same as the promise problem \((S,\complement S)\) (here, \(\complement\) means set complement).

In the following problems, it would be assumed that the input consists of \(\Sigma\), an alphabet, \(G\), a paired de Bruijn graph of order \(k\) over \(\Sigma\), as well as \(1^d\), that is unary coding of \(d\).

\section{Trivial cases}

If \(|\Sigma|=1\), a paired de Bruijn graph can have at most one vertex and at most one edge, and every cycle is sound. If \(k=0\), a paired de Bruijn graph can have at most one vertex and at most \(|\Sigma|^2\) edges, and the problem is a bit harder. However, it can be solved in polymonial time in the following way: construct a directed graph with one vertex for each element of \(\Sigma\) and edge \((u,v)\) iff there is an edge labeled with \((u,v)\) in the original graph (this new graph may contain loops). Now, there is a sound cycle in the original graph iff there is a cycle in the new graph, and there is a covering sound cycle in the original graph iff there is a set of at most \(d\) cycles covering the new graph. Both properties can be easily checked in polymonial time.

\section{A case with fixed \(k\)}

\begin{theorem}\label{fixed-k}
For any fixed \(k\ge1\), the promise problem \((S_+,S_-)\), where \(S_+\) is the set of paired de Bruijn graphs which have a covering sound cycle and \(S_-\) is the set of paired de Bruijn graphs which do not have a sound cycle, is NP-hard.
\end{theorem}

\begin{proof}
The proof of this theorem consists of two parts. Firstly, NP-hardness of a specific graph theory problem is proven by reduction from Hamiltonian Cycle problem. Then, the intermediate problem is reduced to the problem formulated in the theorem.

\begin{lemma}\label{hamiltonian}
The promise problem \((S_+,S_-)\), where \(S_+\) is the set of undirected graphs with a hamiltonian cycle and \(S_-\) is the set of undirected graphs without hamiltonian paths, is NP-hard.
\end{lemma}

\begin{proof}
First note that the problem is well-defined, because every graph with a hamiltonian cycle has a hamiltonian path. We will start with an instance \(G\) of Hamiltonian Cycle problem, which is NP-hard \cite{reducibility}. Without loss of generality, let us assume that \(G\) has at least three vertices.

Now build a new graph \(G'\) in the following way: firstly, pick a vertex in \(G\) and duplicate it together with all the edges incident to it. Let the copies of the vertex be \(a_1\) and \(b_1\). Now let us duplicate the whole graph, let the first copy be \(G_1\) and the second copy be \(G_2\), and let the copies of \(a_1\) and \(b_1\) be \(a_3\) and \(b_3\). Add two new vertices \(a_2\) and \(b_2\) and four new edges \(\{a_1,a_2\}\), \(\{a_2,a_3\}\), \(\{b_1,b_2\}\), and \(\{b_2,b_3\}\) (see \autoref{graph}).

\begin{figure}\centering
\begin{tikzpicture}
	\draw {[minimum size=1.75in]
	      (-1.25in, 0in) node {\(G_1\)}
	      ( 1.25in, 0in) node {\(G_2\)}}
	      {[minimum size=0.5in]
	      (-1in,  0.5in) node (a1) {\(a_1\)}
	      ( 0in,    1in) node (a2) {\(a_2\)}
	      ( 1in,  0.5in) node (a3) {\(a_3\)}
	      (-1in, -0.5in) node (b1) {\(b_1\)}
	      ( 0in,   -1in) node (b2) {\(b_2\)}
	      ( 1in, -0.5in) node (b3) {\(b_3\)}}
	      (a1) -- (a2) -- (a3)
	      (b1) -- (b2) -- (b3);
\end{tikzpicture}
\caption{Graph \(G'\)}\label{graph}
\end{figure}

The following two theorems show that the transformation described above maps all positive instances of hamiltonian cycle problem to \(S_+\) and all negative instances of hamiltonian cycle problem to \(S_-\).

\begin{theorem}
If a graph \(G\) has a hamiltonian cycle, then the graph \(G'\) produced as described above has a hamiltonian cycle.
\end{theorem}

\begin{proof}
After the vertex in \(G\) is duplicated, the cycle in \(G\) maps to a hamiltonian path in \(G_1\) from \(a_1\) to \(b_1\). Analogously, \(G_2\) has a hamiltonian path from \(a_3\) to \(b_3\). So, the cycle in \(G'\) is constructed as follows: start at \(a_1\), traverse the path in \(G_1\) to \(b_1\), go to \(b_2\), then to \(b_3\), then traverse the path in \(G_2\) to \(a_3\), then go to \(a_2\), and return to \(a_1\).
\end{proof}

\begin{theorem}
If a graph \(G\) does not have hamiltonian cycles, then the graph \(G'\) does not have hamiltonian paths.
\end{theorem}

\begin{proof}
Suppose that, on the contrary, \(G'\) contains a hamiltonian path. First consider the case when one end of the path is in \(G_1\) and the other end is in \(G_2\). Then, the path either traverses edges \(\{a_1,a_2\}\) and \(\{a_2,a_3\}\) but not \(\{b_1,b_2\}\) and \(\{b_2,b_3\}\), or \(\{b_1,b_2\}\) and \(\{b_2,b_3\}\) but not \(\{a_1,a_2\}\) and \(\{a_2,a_3\}\). In both cases, either \(a_2\) or \(b_2\) is not visited, so the path is not hamiltonian. So, ends of the path are either both outside \(G_1\), or both outside \(G_2\). Let us assume they are outside \(G_1\), the other case is proved analogously. Besides \(a_1\) and \(b_1\), \(G_1\) contains at least one internal vertex because of the assumption that \(G\) has at least three vertices. To reach that vertex, the path must enter \(G_1\) through \(a_1\) and leave through \(b_1\) (or the opposite, which doesn't matter). Because there are no other ways to enter \(G_1\), the path enters \(G_1\) only once and traverses all vertices of \(G_1\). So, the fragment of the path within \(G_1\), when mapped back to \(G\), becomes a hamiltonian cycle. So, \(G\) has a hamiltonian cycle, a contradiction.
\end{proof}
\end{proof}

\begin{lemma}\label{properties}
If a graph has a hamiltonian cycle, then:
\begin{itemize}
\item For each vertex \(v\) in the graph, there is a hamiltonian path having \(v\) as one of its endpoints.
\item For each edge \(\{u,v\}\) in the graph, there is a hamiltonian path passing through \(\{u,v\}\).
\item For each edge \(\{u,v\}\) and vertex \(w\ne u,v\), there is a hamiltonian path passing through \(\{u,v\}\) such that \(v\) resides between \(u\) and \(w\) on the path.
\end{itemize}
\end{lemma}

\begin{proof}
Let \(n\) be the number of vertices in the graph, and let \(v_1\), \(v_2\), \ldots, \(v_n\) be the vertices numbered in the order of the cycle. Let \(u=v_i\) and \(v=v_j\), \(i<j\) (otherwise, vertices can be renumbered in the reverse order), and let \(w=v_k\). Then, the first point of the theorem is obvious, the path for the second point is \(v_{j+1}\), \(v_{j+2}\), \ldots, \(v_n\), \(v_1\), \(v_2\), \ldots, \(v_i\), \(v_j\), \(v_{j-1}\), \ldots, \(v_{i+1}\), and the path for the third point is the same if \(i<k<j\) and \(v_{j-1}\), \(v_{j-2}\), \ldots, \(v_i\), \(v_j\), \(v_{j+1}\), \ldots, \(v_n\), \(v_1\), \(v_2\), \ldots, \(v_{i-1}\) otherwise.
\end{proof}

Now return to \autoref{fixed-k}. First consider the case \(k=1\). Begin with an instance \(G\) of the problem from \autoref{hamiltonian}. Let \(G\) have \(n\) vertices \(v_1\), \(v_2\), \ldots, \(v_n\). Without loss of generality, let us assume that \(n\ge3\). Set \(d=n+1\). Now, we are going to construct a paired de Bruijn graph \(G'=(V,A)\). It would have block structure: there will be \(2n+2\) blocks \(V_1\), \(V_2\), \ldots, \(V_{2n+2}\) and \(2n+2\) separator vertices \(s_1\), \(s_2\), \ldots, \(s_{2n+2}\), so \(V=V_1\cup V_2\cup\cdots\cup V_{2n+2}\cup\{s_1,s_2,\ldots,s_{2n+2}\}\) (see \autoref{block}). This graph will contain edges of three kinds:

\begin{itemize}
\item Within a block.
\item From \(s_i\) to an element of \(V_i\).
\item From an element of \(V_i\) to \(s_{i+1}\), or from an element of \(V_{2n+2}\) to \(s_1\).
\end{itemize}

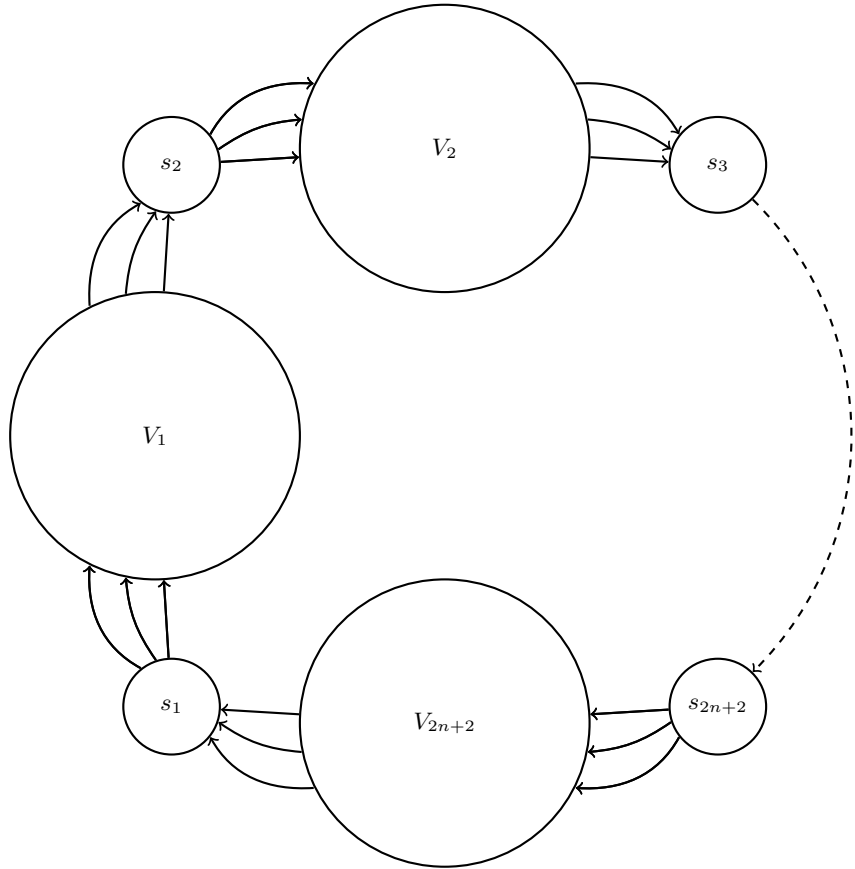
\begin{figure}\centering
\begin{tikzpicture}
	\draw[v/.style={minimum size=1.5in,on chain=placed {at=(-45*\tikzchaincount:1.5in),
	      join=by {out=0,in=180},join=by {out=15,in=165},join=by {out=35,in=150}}},
	      s/.style={minimum size=0.5in,on chain=placed {at=(-45*\tikzchaincount:2in)},
	      join=by {out=0,in=180},join=by {out=15,in=165},join=by {out=30,in=145}},
	      every join/.style={->,relative},start chain]
	      node[s] (sl) {\(s_{2n+2}\)} node[v] (vl) {\(V_{2n+2}\)}
	      node[s] (s1) {\(s_1\)} node[v] (v1) {\(V_1\)}
	      node[s] (s2) {\(s_2\)} node[v] (v2) {\(V_2\)}
	      node[s] (s3) {\(s_3\)} edge[->,dashed,bend left=45] (sl);
\end{tikzpicture}
\caption{Structure of the paired de Bruijn graph}\label{block}
\end{figure}

The alphabet would be analogously divided into \(2n+2\) blocks \(C_1\), \(C_2\), \ldots, \(C_{2n+2}\) and \(2n+2\) separator characters \(t_1\), \(t_2\), \ldots, \(t_{2n+2}\), so \(\Sigma=C_1\cup C_2\cup\cdots\cup C_{2n+2}\cup\{t_1,t_2,\ldots,t_{2n+2}\}\). For each vertex \(v\in V_i\), the first component of the associated bilabel will be in \(C_i\), and the second component will be in \(C_{i+1}\) (or \(C_1\) if \(i=2n+2\)). Each \(s_i\) would be associated with a bilabel \((t_i,t_{i+1})\), \(s_{2n+2}\) will be associated with a bilabel \((t_{2n+2},t_1)\).

The blocks will be formed as follows: the blocks \(V_1\) and \(V_{2n+2}\) would be copies of \(G\), while blocks \(V_2\) through \(V_{2n+1}\) would each contain two copies of \(G\), except for one vertex of which only one copy would be present. The vertices from the first copy would be called \(v_{i,j}\), the vertices from the second copy would be called \(v_{i,j}''\), and the only copy of \(v_{\lfloor i/2\rfloor}\) in block \(i\) would be called \(v_{i,\lfloor i/2\rfloor}'\). The edges would be added such that every path through such block would pass through this vertex.

By assigning a dedicated subset of the alphabet to each block, we prevent vertices from different blocks from being assigned the same bilabel. In fact, it can be seen from the following definition that each vertex is assigned a distinct bilabel, so the assignment is valid. We also note that, with the exceptions of the bilabels containing \(u\), the second index of a character (\(j\) in \(c_{i,j}\)) is the same in both components of a bilabel. This means that in a sound path the sequence of second indices must repeat with a period of \(d\).

The precise definition is as follows:
\begin{itemize}
\item For \(i=1,2n+2\) block \(V_i=\{v_{i,1},v_{i,2},\ldots,v_{i,n}\}\).
\item For \(i=2\ldots2n+1\) block \(V_i=\{v_{i,1},v_{i,2},\ldots,v_{i,\lfloor i/2\rfloor-1},v_{i,\lfloor i/2\rfloor+1},v_{i,\lfloor i/2\rfloor+2},\ldots,v_{i,n},v_{i,\lfloor i/2\rfloor}',v_{i,1}'',v_{i,2}'',v_{i,\lfloor i/2\rfloor-1}'',v_{i,\lfloor i/2\rfloor+1}'',v_{i,\lfloor i/2\rfloor+2}'',\ldots,v_{i,n}''\}\).
\item Alphabet block \(C_1=\{u\}\).
\item For \(i=1\ldots n+1\) alphabet block \(C_{2i}=\{c_{2i,1},c_{2i,2},\ldots,c_{2i,n}\}\).
\item For \(i=1\ldots n\) alphabet block \(C_{2i+1}=\{c_{2i+1,1},c_{2i+1,2},\ldots,c_{2i+1,i-1},c_{2i+1,i+1},c_{2i+1,i+2},\ldots,c_{2i+1,n},c_{2i+1,i}',c_{2i+1,1}'',c_{2i+1,2}'',\ldots,c_{2i+1,i-1}'',c_{2i+1,i+1}'',c_{2i+1,i+2}'',\ldots,c_{2i+1,n}''\}\).
\item For \(i=1\ldots n\) the bilabel associated with \(v_{1,i}\) is \((u,c_{2,i})\).
\item For \(i=1\ldots n\) the bilabel associated with \(v_{2n+2,i}\) is \((c_{2n+2,i},u)\).
\item For \(i=2\ldots2n+1\), \(j=1\ldots n\), \(j\ne\lfloor i/2\rfloor\) the bilabel associated with \(v_{i,j}\) is \((c_{i,j},c_{i+1,j})\).
\item For \(i=1\ldots n\) the bilabel associated with \(v_{2i,i}'\) is \((c_{2i,i},c_{2i+1,i}')\).
\item For \(i=1\ldots n\) the bilabel associated with \(v_{2i+1,i}'\) is \((c_{2i+1,i}',c_{2i+2,i})\).
\item For \(i=1\ldots n\), \(j=1\ldots n\), \(j\ne i\) the bilabel associated with \(v_{2i,j}''\) is \((c_{2i,j},c_{2i+1,j}'')\).
\item For \(i=1\ldots n\), \(j=1\ldots n\), \(j\ne i\) the bilabel associated with \(v_{2i+1,j}''\) is \((c_{2i+1,j}'',c_{2i+2,j})\).
\end{itemize}

The edges are added:
\begin{itemize}
\item For \(i=1\ldots n\) the edges \((s_1,v_{1,i})\), \((v_{1,i},s_2)\), \((s_{2n+2},v_{2n+2,i})\), and \((v_{2n+2,i},s_1)\).
\item For \(i=2\ldots 2n+1\), \(j=1\ldots n\), \(j\ne\lfloor i/2\rfloor\) the edges \((s_i,v_{i,j})\) and \((v_{i,j}'',s_{i+1})\).
\item For \(i=2\ldots 2n+1\) the edges \((s_i,v_{i,\lfloor i/2\rfloor}')\) and \((v_{i,\lfloor i/2\rfloor}',s_{i+1})\).
\end{itemize}
Also, for each edge \(\{v_i,v_j\}\) in \(G\) (\(1\le i\le n\), \(1\le j\le n\), \(i\ne j\)) the following edges are added:
\begin{itemize}
\item The edges \((v_{1,i},v_{1,j})\), \((v_{2n+2,i},v_{2n+2,j})\), \((v_{2j,i},v_{2j,j}')\), \((v_{2j+1,i},v_{2j+1,j}')\), \((v_{2i,i}',v_{2i,j}'')\), and \((v_{2i+1,i}',v_{2i+1,j}'')\).
\item For \(r=2\ldots2n+1\), \(i\ne\lfloor r/2\rfloor\), \(j\ne\lfloor r/2\rfloor\) the edges \((v_{r,i},v_{r,j})\) and \((v_{r,i}'',v_{r,j}'')\).
\end{itemize}
Note that, as edges of \(G\) are undirected, each edge should be processed twice, once as \(\{v_i,v_j\}\) and once as \(\{v_j,v_i\}\). The bilabels associated with the edges can be unambiguously determined from the bilabels associated with their ends.

The size of \(G'\) and the parameter \(d\) is polymonial in terms of \(n\) by construction. The following two theorems show that the transformation described above maps all positive instances of the problem formulated in \autoref{hamiltonian} to paired de Bruijn graphs with covering sound cycles and all negative instances of the problem formulated in \autoref{hamiltonian} to paired de Bruijn graphs without sound cycles.

\begin{theorem}\label{fixed-k-fwd}
If a graph \(G\) has a hamiltonian cycle, then the paired de Bruijn graph \(G'\) produced as described above has a covering sound cycle.
\end{theorem}

\begin{proof}
Remember that \(d=n+1\). Construct the cycle as follows: first, select a hamiltonian path in \(G\), let it be \(v_{p_1}\), \(v_{p_2}\), \ldots, \(v_{p_n}\). Start at \(s_1\), then go to \(v_{1,p_1}\), \(v_{1,p_2}\), \ldots, \(v_{1,p_n}\). Then, for each \(i\) from \(2\) to \(2n+1\), visit \(s_i\), then \(v_{i,p_1}\), \(v_{i,p_2}\), \ldots, \(v_{i,p_{r_{\lfloor i/2\rfloor}-1}}\), where \(r_{\lfloor i/2\rfloor}\) is such that \(p_{r_{\lfloor i/2\rfloor}}=\lfloor i/2\rfloor\), then \(v_{i,p_{r_{\lfloor i/2\rfloor}}}'=v_{i,\lfloor i/2\rfloor}'\), then \(v_{i,p_{r_{\lfloor i/2\rfloor}+1}}''\), \(v_{i,p_{r_{\lfloor i/2\rfloor}+2}}''\), \ldots, \(v_{i,p_n}''\). After that, visit \(s_{2n+2}\), \(v_{2n+2,p_1}\), \(v_{2n+2,p_2}\), \ldots, \(v_{2n+2,p_n}\), and finally return to \(s_1\).

This cycle visits each block, and the sequence of second indices within each block is the same (it is \(p_1\), \(p_2\), \ldots, \(p_n\)), therefore, from the construction, the cycle is sound. However, it is not necessarily covering. To make a covering cycle, first use the procedure described above to construct one cycle per every property from \autoref{properties}, namely:
\begin{itemize}
\item For every vertex \(v_i\), use the path having \(v_i\) as an endpoint to construct cycles passing through \((s_j,v_{j,i})\) (\(1\le j\le 2n+2\), \(i\ne\lfloor j/2\rfloor\)), \((s_{2i},v_{2i,i}')\), \((s_{2i+1},v_{2i+1,i}')\), \((v_{1,i},s_2)\), \((v_{2n+2,i},s_1)\), \((v_{2i,i}',s_{2i+1})\), \((v_{2i+1,i}',s_{2i+2})\), and \((v_{j,i}'',s_{j+1})\) (\(2\le j\le 2n+1\), \(i\ne\lfloor j/2\rfloor\)).
\item For every edge \(\{v_i,v_j\}\), use the path passing through \(\{v_i,v_j\}\) to construct cycles passing through \((v_{r,i},v_{r,j})\) (\(r=1,2n+2\)), \((v_{2j,i},v_{2j,j}')\), \((v_{2j+1,i},v_{2j+1,j}')\), \((v_{2i,i}',v_{2i,j}'')\), and \((v_{2i+1,i}',v_{2i+1,j}'')\).
\item For every edge \(\{v_i,v_j\}\) and vertex \(v_k\) (\(k\ne i,j\)), use the path passing through \(\{v_i,v_j\}\), such that \(v_j\) resides between \(v_i\) and \(v_k\) on the path, to construct cycles passing through \((v_{2k,i},v_{2k,j})\), \((v_{2k+1,i},v_{2k+1,j})\), \((v_{2k,j}'',v_{2k,i}'')\), and \((v_{2k+1,j}'',v_{2k+1,i}'')\).
\end{itemize}
Together, these cycles should cover all the edges of \(G'\). To make a single covering cycle, cut all these cycles at \(s_1\) and join them together. The resulting cycle is sound because the second component of every bilabel from \(V_{2n+2}\) is \(u\), and the first component of every bilabel from \(V_1\) is also \(u\), so they always match.
\end{proof}

\begin{theorem}
If a graph \(G\) doesn't have hamiltonian paths, then the paired de Bruijn graph \(G'\) produced as described above doesn't have sound cycles.
\end{theorem}

\begin{proof}
Within each block, the set of characters used for the first component of bilabels and the set of characters used for the second component of the bilabel do not intersect. Therefore, every contiguous segment of a sound cycle within a single block must have length at most \(d\). Because the blocks are connected in a circle (see \autoref{block}), and the cycle cannot be contained within a single block, it must pass around the circle at least once. Therefore, it must pass through \(s_1\). Exactly \(d\) vertices later, it must pass through \(s_2\), as it is the only vertex with a matching bilabel. The \(d-1=n\) vertices between \(s_1\) and \(s_2\) must be spent within \(V_1\), as the only other way to get to \(s_2\) is to pass around the whole circle at least once, and the circle is longer than \(d\), so this is impossible. Then it must pass through \(s_3\), \(V_3\), \(s_4\), \(V_4\), \ldots, \(s_{2n+2}\), \(V_{2n+2}\), then return to \(s_1\).

Let us call a segment between successive visits to \(s_1\) a pass. Within a pass, each block is visited exactly once, and a path within each block has length \(n\). Moreover, every pair of consecutive blocks, except \((V_{2n+2},V_1)\), has their vertices labeled such that the sequences of second indices within each block must be the same. However, for each \(i\), such that \(1\le i\le n\), the structure of blocks \(V_{2i}\) and \(V_{2i+1}\) requires the sequence of second indices to include \(i\), as it it impossible to pass though these blocks otherwise. Therefore, the sequence must include every value from \(1\) to \(n\), so it is a permutation. Since every edge in \(G'\) within a block corresponde to an edge in \(G\), the permutation defines a hamiltonian path in \(G\), a contradiction.

The case \(k>1\) is handled as follows: first, produce a graph \(G\) over an alphabet \(\Sigma\) for the case \(k=1\). Then, construct a new alphabet \(\Sigma'\) as being equal to \(\Sigma\cup\{f\}\), where \(f\) is a new character. After that, construct a new graph \(G'\) from \(G\) by replacing each vertex labeled \((a,b)\) with \(k'\) vertices labeled \((f^{k'-1}a,f^{k'-1}b)\), \((f^{k'-2}af,f^{k'-2}bf)\), \ldots, \((af^{k'-1},bf^{k'-1})\) and \(k'-1\) edges labeled \((f^{k'-1}af,f^{k'-1}bf)\), \((f^{k'-2}aff,f^{k'-2}bff)\), \ldots, \((faf^{k'-1},fbf^{k'-1})\), and replacing each edge labeled with \((ab,cd)\) with an edge labeled \((af^{k'-1}b,cf^{k'-1}d)\). Finally, set \(d'\) equal \(k'd\). Now, every sound cycle in can be unambiguously mapped from \(G\) to \(G'\) and vice versa. Therefore, the new solution is equivalent to the old one.
\end{proof}
\end{proof}

These immediately follow from \autoref{fixed-k}:

\begin{corollary}
The problem of checking whether a paired de Bruijn graph contains a sound cycle is NP-hard, both in general case and for any fixed \(k\ge1\).
\end{corollary}

\begin{corollary}
The problem of checking whether a paired de Bruijn graph contains a covering sound cycle is NP-hard, both in general case and for any fixed \(k\ge1\).
\end{corollary}

\section{A case with fixed \(|\Sigma|\)}

\begin{theorem}\label{fixed-sigma}
For any fixed \(|\Sigma|\ge2\), the promise problem \((S_+,S_-)\), where \(S_+\) is the set of paired de Bruijn graphs which have a covering sound cycle and \(S_-\) is the set of paired de Bruijn graphs which don't have a sound cycle, is NP-hard.
\end{theorem}

\begin{proof}
This is proven by reduction from the same problem with fixed \(k=1\). Let the instance with \(k=1\) be \(G\), and let its alphabet be \(\Sigma\). We are going to build an instance \(G'\) of the same problem with alphabet \(\Sigma'=\{0,1\}\). Set \(l=\lceil\log_2|\Sigma|\rceil\). Now, every character from \(\Sigma\) can be unambiguously encoded with \(l\) binary digits. Take that encoding, and replace each digit \(0\) with the sequence \(01\), and each digit \(1\) with the sequence \(10\). The resulting encoding of length \(2l\) has the following properties: it does not contain repetitions of three or more of the same digit as a substring, and it does not begin or end with a repeated digit. Set \(k'=4l+5\). Let \(\enc(c)\) denote the \(2l\)-character encoding of \(c\) described above. Then, for each vertex in \(G\) labeled \((a,b)\), add a vertex labeled \((\enc(a)01110\enc(a),\enc(b)01110\enc(b))\). Here, the sequence \(111\) unambiguously determines the center of the encoding of a character. Each edge from \(G\) is translated to \(4l+9\) new vertices and \(4l+10\) new edges: if the original edge has the bilabel \((ab,cd)\), the bilabels of the new vertices and edges will spell \((\enc(a)01110\enc(a)10001\enc(b)01110\enc(b),\enc(c)01110\enc(c)10001\enc(d)01110\enc(d))\). Each bilabel would include at least one of the marker sequences \(000\) and \(111\) and at least one complete encoding of a character, so there will be no undesired overlaps. It can be shown that each sound cycle from \(G\) can be mapped to \(G'\) and vice versa, so they are equivalent for the purposes of the problem.
\end{proof}

These immediately follow from \autoref{fixed-sigma}:

\begin{corollary}
The problem of checking whether a paired de Bruijn graph contains a sound cycle is NP-hard for any fixed \(|\Sigma|\ge2\).
\end{corollary}

\begin{corollary}
The problem of checking whether a paired de Bruijn graph contains a covering sound cycle is NP-hard for any fixed \(|\Sigma|\ge2\).
\end{corollary}

\section{A case with both \(k\) and \(|\Sigma|\) fixed}

If both \(k\) and \(|\Sigma|\) are fixed, the number of possible paired de Bruijn graphs is limited: there are at most \(|\Sigma|^{2k}\) different vertex bilabels, and at most \(|\Sigma|^{2k+2}\) different edge bilabels, and each bilabel is used by at most one vertex or edge, so the total number of different paired de Bruijn graphs is limited by a number which only depends on \(k\) and \(|\Sigma|\). Let us denote this number by \(N\).

There are at most \(N\) different problem instances for each instance length: otherwise, there would be two different instances having the same graph and the same length, but such instances can only differ in \(d\), which is represented in unary coding, so any instances which only differ in \(d\) must have different length. Therefore, the number of instances is polynomial in instance length, so the language defined by the problem is \emph{sparse}. Unless P=NP, a sparse language is never NP-hard \cite{sparse}. Therefore, the problem of checking whether a paired de Bruijn graph has a sound cycle cannot be NP-hard if both \(k\) and \(|\Sigma|\) are fixed.

\section{Conclusion}

We have proved that the Paired de Bruijn Sound Cycle problem is NP-hard in general case. Results of this work combined with previous works on genome assembly complexity show that all known models for genome assembly both from single and mate-pair reads are NP-hard.

However, the problem considered in this paper has a special case with both $k$ and $|\Sigma|$ fixed which is not NP-hard unless P=NP. A reasonable direction of future research is to determine if this case is solvable in polynomial time.

\bibliographystyle{splncs03}
\bibliography{\jobname}

\end{document}